\documentclass[10pt, conference, letterpaper]{IEEEtran}
\IEEEoverridecommandlockouts

\newtheorem{theorem}{Theorem}
\newtheorem{assumption}{Assumption}
\newtheorem{lemma}{Lemma}

\newtheorem{example}{Example}

\usepackage{siunitx}
\usepackage{comment}
\usepackage{cite}
\usepackage{amsmath,amssymb,amsfonts}

\usepackage{multirow}
\usepackage{algorithm}
\usepackage{algpseudocode}
\usepackage{graphicx}
\usepackage[caption=false,font=footnotesize]{subfig}
\usepackage{textcomp}
\usepackage{xcolor}
\def\BibTeX{{\rm B\kern-.05em{\sc i\kern-.025em b}\kern-.08em
    T\kern-.1667em\lower.7ex\hbox{E}\kern-.125emX}}

\begin{document}

\title{ImpactHO: Importance-Aware KV Cache Transfer for Multi-User Edge LLM Handover}

\author{
    \IEEEauthorblockN{Minwoo Kim, Soochang Song, Namyoon Lee, Bang Chul Jung, and Yongjune Kim}
 
    \thanks{
        M. Kim, S. Song, N. Lee, and Y. Kim are with the Department of Electrical Engineering, Pohang University of Science and Technology (POSTECH), Pohang 37673, South Korea (e-mail: \{minwoo.kim, ssc6351, nylee,  yongjune\}@postech.ac.kr)
        
        B. C. Jung is with the Department of Electrical and Computer Engineering, Ajou University, Suwon 16499, South Korea (e-mail: bcjung@ajou.ac.kr).
    }
}

\maketitle

\begin{abstract}
    Edge LLMs must preserve inference continuity when a user hands over between edge nodes, requiring key-value (KV) cache transfer to the target node.
    However, simultaneous handovers saturate the backhaul, preventing full cache delivery within the mobility-imposed transfer window.
    Rather than allocating bandwidth as if all cache entries were equally valuable, we order each user's KV cache by importance and transmit only its most informative fraction, turning token-level sparsity into communication savings.
    We cast the transfer as a multi-user backhaul allocation problem that maximizes average accuracy across users.
    Each user's partial-cache accuracy serves as its utility: a sigmoid that fits measurements on the RULER benchmark with $R^2 > 0.99$ across models and context lengths.
    Because importance ordering front-loads the high-value entries, the concave region of the accuracy curve spans nearly the entire cache.
    Our proposed allocator keeps served users within this region, making each per-slot allocation problem convex.
    The optimum is derived via a weighted water-filling solution that generalizes information-theoretic water-filling and enables online scheduling.
    The proposed allocator attains over 93.7\% average accuracy in a 500ms transfer window, within 0.5pp of the full-cache ceiling, and reaches 98.2--99.5\% of a clairvoyant upper bound.
\end{abstract}


\section{Introduction}
    Next-generation communication networks are evolving beyond simple data delivery into an infrastructure that supports real-time artificial intelligence (AI) inference for applications such as large language model (LLM) services, autonomous driving, and robotics~\cite{Saad2019vision, Letaief2019roadmap}.
    These services share two requirements: ultra-low latency and seamless mobility.
    The first is difficult to meet with remote cloud inference, whose round-trip time alone can exceed the millisecond-scale latency budget of real-time AI~\cite{Mao2017survey}.
    Multi-access edge computing (MEC) addresses this by hosting computation at the network edge, close to users.
    For AI workloads, this paradigm has taken shape as Edge AI and, more recently, Edge LLM, in which model inference runs directly on edge nodes~\cite{Li2020edge, Kholmatov2025AoRA}.

    However, unlike cloud-based serving, edge inference is directly exposed to user mobility.
    Whenever a moving user is reassigned to another edge node, inference continuity must be preserved across the handover.
    Although virtual machine or container migration has been studied for MEC handover~\cite{Wang2019dynamic, Machen2018live, Ngo2020coordinated}, modern AI models carry tens to hundreds of gigabytes of weights (e.g., \SI{16}{GB} for 8B-parameter models), making such migration impractical and pushing handover latency to tens of seconds.

    Since the model itself is already provisioned at each edge node, only the user-specific context needs to move.
    Transformer-based LLMs store this context as a key-value (KV) cache, and recent work proposes treating the KV cache itself as the transfer target~\cite{Lee2026low-latency}.
    Unlike transferring the raw text or tokens, which requires re-prefilling at the target node, transferring the KV cache lets the target node resume inference immediately.
    However, the KV cache is substantially larger than the raw text it encodes, so concurrent long-context handovers still saturate the backhaul.   

    \begin{figure}[t] 
        \centering
        \includegraphics[width=0.42\textwidth]{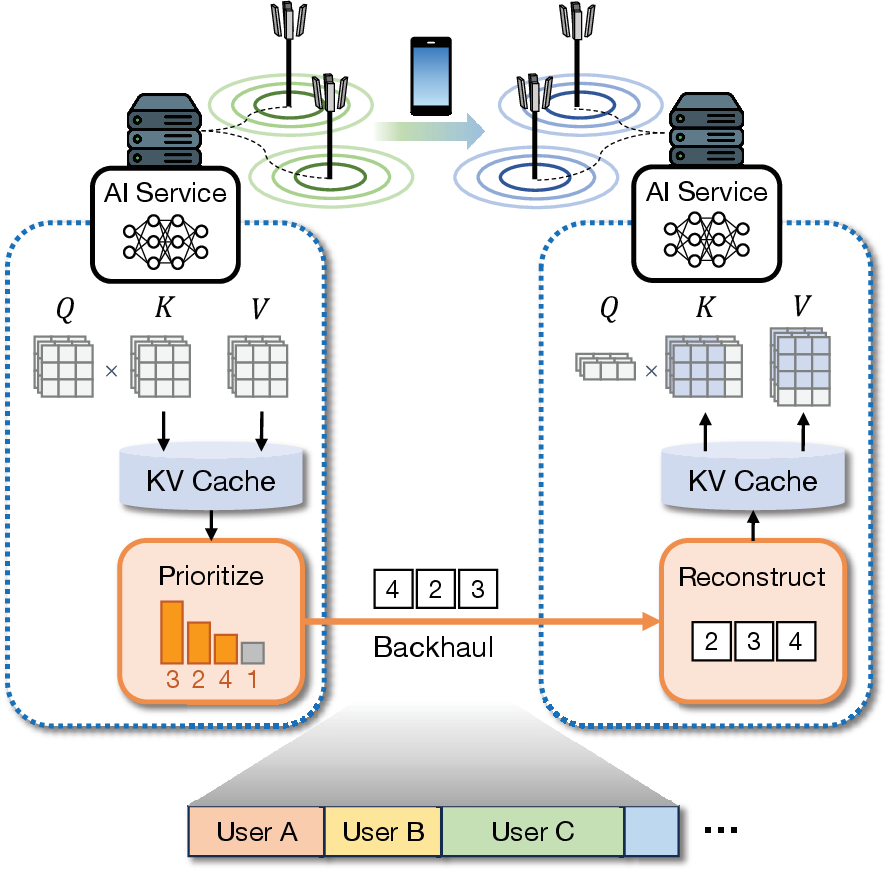}
        \caption{Concept of ImpactHO framework.}
        \label{fig:concept}
        \vspace{-4mm}
    \end{figure}

    A natural remedy is to transfer only part of the KV cache.
    Recent studies on KV cache eviction have shown that not all entries are equally important at inference time: LLMs can retain most of their accuracy using only a small fraction of their cache~\cite{Xiao2024efficient, Zhang2023H2O, Li2024snapkv, Kim2025kvzip, Kim2026fastkvzip}.
    These results, however, address a single model's memory budget, not how a shared backhaul should be divided across users.

    To address this gap, we propose \textbf{ImpactHO} (\textbf{Imp}ortance-\textbf{a}ware KV \textbf{c}ache \textbf{t}ransfer for multi-user \textbf{h}and\textbf{o}ver), a framework that orders each user's cache by importance and transmits its most informative entries first, as illustrated in Fig.~\ref{fig:concept}.    
    To capture how accuracy grows with the delivered cache, we model each user's accuracy as a function of its received cache fraction and treat this curve as a per-user utility. 
    The scheduling task then becomes a utility-maximization problem: allocate the limited backhaul bandwidth across users so as to maximize the average accuracy. 
    
    Importance ordering plays a structural role beyond prioritizing informative cache.
    Partial-cache accuracy is intrinsically concave once enough context has arrived for the model to become operable. 
    By front-loading the highest-value entries, importance ordering pulls this operable point (i.e., the inflection point of the accuracy curve) down to only a small cache fraction.
    The concave region therefore spans nearly the entire cache, and every user can be brought into it at low cost.
    An admission rule then lifts every served user past this point, so each per-slot allocation reduces to a convex subproblem, whose optimum takes a weighted water-filling form with each user's cache size acting as its weight.
    
    We summarize our main contributions as follows:
    \begin{itemize}
        \item \emph{ImpactHO framework}:
        We formulate importance-ordered partial KV cache transfer as a multi-user backhaul allocation problem for edge LLM handover, repurposing per-entry importance scores from KV cache eviction to set the transmission order.
        Modeling each user's partial-cache accuracy as its utility casts the scheduling task as a utility-maximization problem over the shared bandwidth.
    
        \item \emph{Empirical sigmoid characterization of partial-cache accuracy}:
        On the RULER benchmark~\cite{Hsieh2024ruler}, partial-cache accuracy follows a sigmoid ($R^2>0.99$) robustly across context lengths, models, and ordering schemes.
        Importance ordering pulls its inflection point down to about \SI{6.5}{\%} of the cache, so the concave region spans nearly the entire cache.
        This validates our utility assumption and provides a reusable parametric foundation for KV cache-aware networking research.        
    
        \item \emph{Two-regime allocator with low overhead}:
        We derive the per-slot optimum in closed form as a weighted water-filling solution over a feasible region expanded by importance ordering.
        It reduces to classical water-filling as a special case and runs in near-linear time per slot.
        Coupled with an admission rule that suppresses starvation under heavy load, the allocator consistently outperforms baselines and reaches near-full accuracy faster than target-side re-prefill of even a single 8K context.
    \end{itemize}

\section{Related Work}\label{sec:related}
    The KV cache stores the key and value tensors of previously processed tokens at each Transformer layer and attention head, avoiding their recomputation during autoregressive decoding.
    However, because its memory footprint grows linearly with the context length, the KV cache has become a major bottleneck for long-context inference.    
    To alleviate this bottleneck, eviction-based methods~\cite{Xiao2024efficient, Zhang2023H2O, Li2024snapkv, Kim2025kvzip, Kim2026fastkvzip} estimate the importance of cached entries and retain only the most important ones, thereby reducing memory consumption while largely preserving model accuracy.

    Among these, KVzip~\cite{Kim2025kvzip} derives query-agnostic per-(layer, head, token) importance scores from a context-reconstruction pass, reaching near-lossless accuracy with less than \SI{30}{\%} of the cache.    
    Fast KVzip~\cite{Kim2026fastkvzip} distills these scores into lightweight gating modules offline, achieving comparable accuracy without the inference-time scoring overhead.
    While prior work uses these importance scores only for memory reduction, we repurpose them to set the transmission order over the backhaul.
    Since the future query is unavailable at handover time, query-agnostic importance estimation is particularly well suited to our setting; hence, we adopt Fast KVzip.
    Nevertheless, our framework can accommodate any scoring method that provides a per-entry importance ranking.

    Beyond the memory pressure addressed by eviction, moving the cache between nodes is itself a bottleneck. 
    In disaggregated datacenter serving, DistServe~\cite{Zhong2024distserve} and Splitwise~\cite{Patel2024splitwise} place prefill and decode on separate instances, requiring the KV cache to be transferred between them.
    However, these systems are designed for datacenter environments with high-bandwidth interconnects and do not address mobility-driven handovers or contention among multiple migrating users over bandwidth-constrained backhaul links. 
    More closely related to our setting, CacheGen~\cite{Liu2024cachegen} streams a compressed KV cache adaptively, but focuses on loading a single context rather than allocating shared backhaul bandwidth across concurrent handovers. 
    Such KV compression is complementary to our allocation framework: it can be applied before transmission to further reduce the required backhaul traffic. 
    
    While sharing our motivation to preserve inference continuity under mobility, a recent edge LLM handover scheme ctHO~\cite{Lee2026low-latency} adopts a fundamentally different formulation. 
    It minimizes the maximum handover delay across users, couples resource allocation with target-side computation, and requires the handover timing of all users to be known in advance.
    In contrast, we transmit only a selected fraction of each KV cache and maximize inference accuracy subject to an online per-slot backhaul budget, thereby decoupling backhaul allocation from target-side compute.
    ImpactHO supports \emph{anytime} inference because the target can resume decoding after any completed prefix of the KV cache stream using the partial cache received up to that point. 
    In contrast, re-prefill must complete prefill over the full context, while ctHO must complete its prescribed cache-transfer and target-side re-prefill operations before decoding can resume. 
    Hence, neither baseline supports inference before full context restoration.
    Table~\ref{tab:structural} summarizes these structural differences, and Section~\ref{subsec:exp-comparison} provides a quantitative comparison.

    \begin{table}[t]
    \centering
    \small
    \renewcommand{\arraystretch}{1.2}
    \caption{Structural comparison of handover strategies.}
    \label{tab:structural}
    \begin{tabular}{|l|c|c|c|}
        \hline
        & Re-prefill
        & ctHO~\cite{Lee2026low-latency}
        & \textbf{ImpactHO} \\
        \hline\hline
        Target-GPU prefill     & Full            & Partial       & \textbf{None} \\ \hline
        Scheduling             & N/A             & Offline       & \textbf{Online} \\ \hline
        Anytime property       & No              & No            & \textbf{Yes} \\
        \hline
    \end{tabular}
    \vspace{-4mm}
    \end{table}

    Our importance-ordered KV cache transfer shares the principle of task-oriented and semantic communications, which prioritize information according to its relevance to the given task rather than bit-level fidelity~\cite{Gunduz2022beyond, Im2024attention}.
    A recent instantiation of this principle is token communications, in which tokenized multimodal source signals are transmitted and reconstructed at the receiver~\cite{Qiao2025token, Qiao2025token-domain}.
    Our setting differs in what is transmitted: not source tokens, but the per-token KV cache entries constituting the internal inference state of a deployed LLM.
    The target node consumes these entries directly to resume inference, without reconstructing the original source. 
    By sending only a high-value fraction in descending order of measured importance, we extend this paradigm to mainstream LLM serving.
    
    Our framework also relates to network utility maximization (NUM), which shares a fixed capacity by maximizing the sum of per-user utilities.
    For elastic flows with concave utilities, NUM admits globally optimal distributed solutions~\cite{Kelly1998rate}.
    Real-time inelastic flows are a closer match to our setting: their utilities are sigmoidal, and this shape makes the problem nonconvex, so standard dual algorithms can miss a feasible global optimum.
    Prior work has addressed this regime in two ways: sub-threshold flows can self-regulate, turning off when a persistently low net utility signals infeasibility~\cite{Lee2005non-convex}; alternatively, link capacity can be provisioned large enough to guarantee convergence of the distributed algorithm~\cite{Hande2007distributed}.

    In both approaches, the application determines a fixed utility curve, and the network can only select an operating point along that curve.
    In our setting, by contrast, the utility is empirically measured and depends not only on the user's context length but also on the order in which KV cache entries are transmitted.
    Our approach is therefore to shape the utility itself through importance ordering, so that the concave regime covers nearly the entire operating range.
    
    \section{System Model and Problem Formulation}\label{sec:system}

    \subsection{Framework Overview}\label{subsec:framework}
    
    The overall concept of ImpactHO is illustrated in Fig.~\ref{fig:concept}.
    While serving a user, the source node assigns an importance score to each KV cache entry and maintains the entries in descending order of importance.
    Consequently, any prefix of the ordered cache contains the highest-ranked entries for the corresponding cache size.
    When a handover occurs, the source node streams these entries to the target node in importance order over the shared backhaul.
    If the finite transfer window expires before the entire cache is delivered, this ordering ensures that the target has received the most valuable subset available under the transfer budget and can resume inference using the resulting partial cache.    
    When multiple handovers concurrently contend for the shared backhaul, ImpactHO allocates the per-slot transmission budget among users to maximize their aggregate inference accuracy.
    The remainder of this section formalizes the KV cache size and importance-induced accuracy utility and formulates the resulting per-slot resource-allocation problem.

    \subsection{Multi-user Edge LLM Handover}\label{subsec:system-multi}
    
    We consider an edge node hosting a shared LLM that concurrently serves multiple users, reflecting a typical edge deployment in which GPU memory constraints favor sharing a single model instance across users. 
    Each user $i$ maintains an individual context stored at the edge node as a KV cache of size 
    \begin{equation} \label{eq:kv_size}
        L_i = 2 n_L n_H d_h q T_i \quad \text{bits},
    \end{equation}
    where the factor of two accounts for the key and value tensors, $n_L$ is the number of transformer layers, $n_H$ is the number of KV heads, $d_h$ is the head dimension, $q$ is the number of bits used to represent each scalar, and $T_i$ is the context length.
    $L_i$ depends on the transformer architecture and grows linearly with the context length $T_i$.
    For example, the KV cache of Qwen3-8B~\cite{Yang2025qwen3} occupies approximately \SI{1.2}{GB}~(\SI{9.66}{Gb}) for an 8K-token context.
    Once a handover is triggered, we assume that $L_i$ remains fixed throughout the KV cache transfer.
    
    An edge inference node may serve users across multiple radio cells; hence, inference-state migration is required only when a user moves beyond the coverage of its current node.
    To preserve inference continuity, the source node transfers the user's KV cache to the target node.
    In AI-RAN architectures for beyond-5G networks, edge inference nodes may be deployed at the distributed unit (DU) or central unit (CU) level and connected through fiber-based transport networks using interfaces such as F1 or X2/Xn~\cite{Kundu2025AI-RAN}.
    For brevity, we collectively refer to the transport path between the source and target nodes as the \emph{backhaul}.
    KV cache transfers share a backhaul bandwidth of $B$ bits/s among the active handover users, where $B$ denotes the bandwidth allocated to KV cache transfer, not the raw physical link capacity.

    We model backhaul allocation in discrete time slots of duration $\Delta t$.
    At the beginning of each slot, the scheduler observes the set of active handover users and determines their allocations, which remain fixed throughout the slot.
    A handover request arriving during a slot becomes eligible for scheduling at the beginning of the next slot.
    We focus on the given slot in which $N$ users concurrently undergo handover and compete for the available backhaul bandwidth $B$, i.e., per-slot budget of $B \Delta t$ bits.
    We denote by $\mathcal{N}_t$ the index set of these active handover users, so that $|\mathcal{N}_t| = N$.
    We consider a single transfer direction, with $B$ the bandwidth provisioned for that direction; the $N$ contending users are those handing over in it, and the reverse direction forms a symmetric instance.
    
    \subsection{Importance-aware KV Cache Ordering}\label{subsec:system-ordering}
    
    We assign an importance score to each individual KV cache entry, defined as the key--value pair associated with one token at a specific layer and KV head, rather than to an entire token.    
    Each entry occupies $2d_hq$ bits (\SI{0.5}{KB} for Qwen3-8B in BF16), and user $i$'s cache therefore contains $n_L n_H T_i$ entries.
    The source node orders these entries by decreasing importance and transmits them in that order.
    Consequently, any delivered prefix is the top-ranked subset of its size: when a fraction $x_i\in[0,1]$ of user $i$'s cache has arrived, the target node holds the highest-ranked $x_i L_i$ bits.
    As $x_i$ increases, newly transmitted entries have progressively lower importance, naturally inducing diminishing returns.
    This observation motivates the concave accuracy-utility model introduced in Section~\ref{subsec:system-utility}.

    In this paper, we obtain the importance ordering from the gating-network scores of Fast KVzip~\cite{Kim2026fastkvzip}, chosen for its near-state-of-the-art retention quality at low computational cost.
    Importance-ordered transmission does introduce implementation overhead beyond standard cache transfer, which we quantify and discuss in Section~\ref{subsec:exp-discussion}.
    Our optimization framework is not restricted to Fast KVzip or even to importance-based ordering: any fixed ordering, including random ordering, can be accommodated as long as its induced utility is monotonically increasing and concave over the operating region considered; Section~\ref{sec:utility} examines this condition in detail.    

    \subsection{Utility Function} \label{subsec:system-utility}
    
    We define the per-user utility $A_i(y)$ as the inference accuracy for user $i$ when the target node has received a fraction $y \in [0, 1]$ of the user's importance-ordered KV cache. 
    Once the target node has received a sufficient fraction, increasing $y$ adds progressively lower-ranked entries to those already delivered, so accuracy is expected to improve with diminishing marginal gains.
    This motivates the following assumption, on which the optimal-allocation analysis of Section~\ref{sec:importance-aware} relies.    

    \begin{assumption}[Concave Operating Region]
    \label{asm:operable}
    For each user $i$, there exists a concavity anchor $\tau_i \in (0,1)$ such
    that $A_i(y)$ is continuously differentiable, strictly increasing, and strictly concave on $[\tau_i,1]$. 
    We restrict the analysis to $y \in [\tau_i,1]$. 
    \end{assumption}

    Operationally, $\tau_i$ denotes the minimum cache fraction at which the target enters the concave operating region considered for resource allocation.
    Below $\tau_i$, the received KV cache entries may be insufficient for reliable task performance, and $A_i$ need not be concave.
    Above $\tau_i$, transmitting additional, progressively lower-ranked entries yields diminishing accuracy gains.
    Section~\ref{sec:utility} shows that the measured accuracy curves exhibit operating regions consistent with Assumption~\ref{asm:operable}.
    
    \subsection{Optimization Problem}\label{subsec:system-optimization}
    At the beginning of a slot, the target node has received a fraction $x_i$ of user $i$'s KV cache.
    The scheduler assigns user $i$ a \emph{normalized} cache-delivery rate $b_i$, in cache fractions per second, so that the received fraction at the end of the slot is
    \begin{equation}
        y_i \triangleq x_i + b_i\Delta t.
    \end{equation}
    Because the full cache contains $L_i$ bits, this allocation corresponds to a transmission rate of $L_i b_i$ bits/s.    
    
    We maximize the aggregate accuracy improvement achieved during the slot:
    \begin{equation}
        \label{eq:objective}
        \begin{aligned}
        \operatorname*{maximize}_{\{b_i\}}\quad & \sum_{i=1}^{N}\bigl[\,A_i(y_i) - A_i(x_i)\,\bigr]\\
        \text{subject to}\quad
        & \sum_{i=1}^{N} L_i\, b_i \le B, \\
        & b_i \ge 0,\quad \tau_i \le y_i \le 1,\quad \forall\, i \in \{1,\dots,N\}.
        \end{aligned}
    \end{equation}    
    The first constraint enforces that the aggregate transmission rate does not exceed the available backhaul bandwidth $B$.
    The remaining constraints prevent the received cache fraction from decreasing, restrict each user to the concave operating region specified in Assumption~\ref{asm:operable}, and cap the received fraction at the full cache.
    Because $b_i\ge0$ is equivalent to $y_i\ge x_i$, the two lower bounds on $y_i$ can be combined as $\tilde{x}_i \triangleq \max\{x_i,\tau_i\}$, yielding the individual feasible interval $y_i\in[\tilde{x}_i,1]$.

    Moreover, $x_i$ is fixed at the beginning of the slot, so $\sum_i A_i(x_i)$ is constant and can be omitted from the objective. 
    Using $b_i=(y_i-x_i)/\Delta t$, Problem~\eqref{eq:objective} is equivalent to    
    \begin{equation}
    \label{eq:obj-std}
    \begin{aligned}
    \operatorname*{minimize}_{\{y_i\}}\quad
        & -\sum_{i=1}^{N} A_i(y_i) \\
    \text{subject to}\quad
        & \sum_{i=1}^{N} \frac{L_i(y_i-x_i)}{\Delta t} \le B, \\
        & \tilde{x}_i \le y_i \le 1,
          \quad \forall i\in\{1,\ldots,N\}.
    \end{aligned}
    \end{equation}
    The objective over the feasible set is convex and all constraints are affine. 
    We solve this convex optimization problem in Section~\ref{sec:importance-aware}.

    The formulation requires every active user to reach its concave operating region by the end of the slot. This is possible only when the available backhaul budget is sufficient to raise every user currently below its concavity anchor $\tau$.

    \begin{lemma}[Feasibility]
    \label{lem:feasibility}
    Problem~\eqref{eq:obj-std} is feasible if and only if
    \begin{equation}
    \label{eq:bmin}
    B \;\ge\; B_{\min} \;\triangleq\;
    \frac{1}{\Delta t}\sum_{i=1}^{N} L_i\,[\tau_i - x_i]^{+},
    \end{equation}
    where $[\,\cdot\,]^{+} \triangleq \max\{\cdot,\,0\}$.
    \end{lemma}
    \begin{IEEEproof}
    Any feasible $\{y_i\}$ satisfies $y_i \ge \tilde{x}_i$, hence $y_i-x_i \ge \max\{x_i,\tau_i\}-x_i = [\tau_i-x_i]^+.$
    It must therefore transmit at least $\sum_iL_i[\tau_i-x_i]^+$ bits during the slot, proving the necessity of~\eqref{eq:bmin}.
    Conversely, if~\eqref{eq:bmin} holds, choosing $y_i=\tilde{x}_i$ for every user satisfies both the individual bounds
    and the backhaul constraint.    
    \end{IEEEproof}

    When $B<B_{\min}$, not all active users can reach their concave operating regions within the slot, and we instead invoke the fallback policy described in Section~\ref{sec:importance-aware}.

    \section{Utility Function Characterization}\label{sec:utility}
    
    Recall from Section~\ref{subsec:system-utility} that $A_i(y)$ denotes the inference accuracy for user $i$ as a function of its received cache fraction $y$, and that Assumption~\ref{asm:operable} posits a concavity anchor $\tau_i$ beyond which $A_i(y)$ is concave. 
    We hypothesize that $A_i(y)$ exhibits a sigmoidal shape: accuracy remains low when the received cache is insufficient, increases rapidly once enough informative entries have been delivered, and eventually saturates because subsequently delivered entries are progressively lower ranked and provide smaller marginal gains.
    This section evaluates this hypothesis through systematic measurements.
    In addition to supporting the analysis in Section~\ref{sec:importance-aware}, the resulting characterization provides a compact parametric utility model for importance-ordered KV cache transfer.
    
    \subsubsection{Empirical Setup}
    We test this hypothesis by fitting four canonical sigmoid families---algebraic, logistic, error function (erf), and arctan---to measured LLM accuracy across multiple context lengths.
    Although prior KV cache eviction studies have reported the qualitative dependence of LLM accuracy on the retained cache fraction~\cite{Zhang2023H2O, Kim2026fastkvzip}, our objective is to identify a parametric model that both accurately represents the measured utility and permits tractable resource-allocation analysis.
    We use Qwen3-8B~\cite{Yang2025qwen3} as the primary LLM, and additionally evaluate Qwen3-14B and Llama-3.1-8B-Instruct~\cite{Grattafiori2024llama3} to assess robustness across model architectures and sizes.
    As the benchmark, we use RULER~\cite{Hsieh2024ruler}, a long-context evaluation suite comprising diverse tasks at multiple context lengths.
    We compare the importance ordering produced by Fast KVzip with random ordering as a baseline.
    
    \subsubsection{Fitting Results}
    
    \begin{figure}[t]
        \centering
        \subfloat[Fast KVzip, 8K context]{\includegraphics[width=0.48\linewidth]{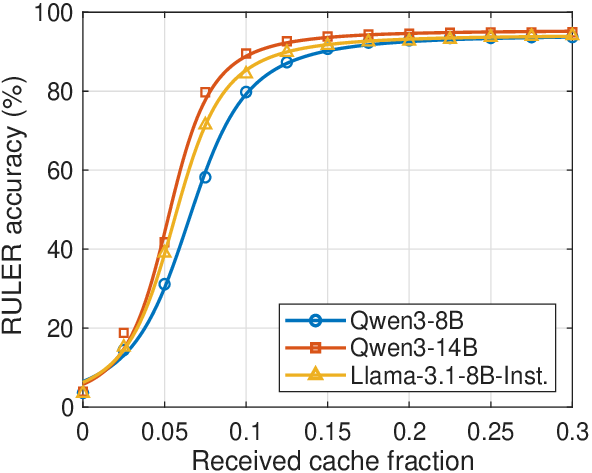}\label{fig:fit-8K}}
        \hfill
        \subfloat[Random ordering, 8K context]{\includegraphics[width=0.48\linewidth]{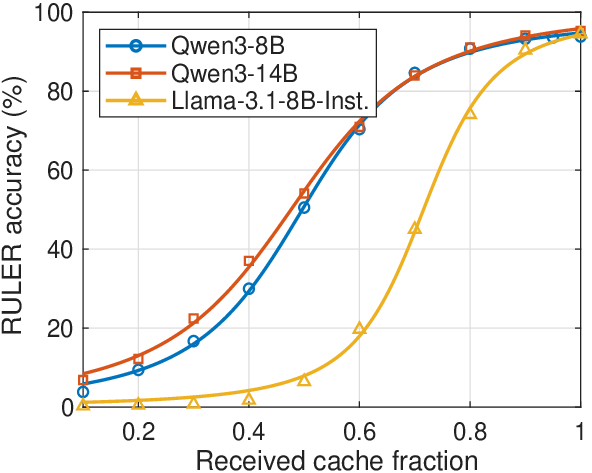}\label{fig:fit-random}}
        \caption{Algebraic sigmoid fits to the measured RULER accuracy as a function of the received KV cache fraction at 8K context length for Qwen3-8B, Qwen3-14B, and Llama-3.1-8B-Instruct: (a) Fast KVzip and (b) random ordering.
        Markers indicate measured RULER accuracy, and solid curves show the algebraic sigmoid fits defined in~\eqref{eq:algebraic}.
        }
        \label{fig:fit}
        \vspace{-4mm}
    \end{figure}
    
    Fig.~\ref{fig:fit} shows that the algebraic sigmoid fits well for Qwen3-8B at 8K ($R^2>0.999$, RMSE $=0.936$).
    The algebraic sigmoid remains accurate across the 4K, 8K, and 16K context lengths, achieving $R^2>0.999$ in every case. 
    Thus, the observed sigmoidal behavior is not specific to a single context length.  
    The remaining three families fit the same measurements equally well, so the sigmoidal shape is a property of the measured accuracy.
    The fitted inflection points are stable across the four functional families, lying within the narrow range $\tau\in[0.064,0.067]$. 
    Moreover, the algebraic sigmoid achieves $R^2>0.99$ for both Qwen3-14B and Llama-3.1-8B-Instruct under both Fast KVzip and random ordering. 
    Collectively, these results support the existence of a concave operating region consistent with Assumption~\ref{asm:operable} across the considered models, context lengths, and transmission orderings.
        
    \subsubsection{Choice of Functional Form}
    
    A practical utility model should provide both an accurate empirical fit and a tractable marginal-utility inversion.
    Although all four fitted families yield optimal solutions, we adopt the \emph{algebraic} sigmoid because its marginal-utility inverse involves only arithmetic operations and radicals:
    \begin{align}\label{eq:algebraic}
        A_i(y) &= \frac{M_i}{2}\!\left(\,
          1 + \frac{k_i(y - \tau_i)}{\sqrt{1 + k_i^2(y - \tau_i)^2}}\, \right),\\   \label{eq:algebraic-deriv}
        A_i'(y) &= \frac{M_i\, k_i}
          {2\,\bigl(1 + k_i^2(y - \tau_i)^2\bigr)^{3/2}}.
    \end{align}
    
    The parameters $M_i$, $k_i$, and $\tau_i$ represent the upper accuracy asymptote, transition sharpness, and concavity anchor, respectively. 
    These parameters may vary across users because of differences in context length and task characteristics.

    \section{Importance-Aware Resource Allocation via Weighted Water-Filling}\label{sec:importance-aware}

    \subsection{Optimal Weighted Water-Filling Structure}

    The solution to Problem~\eqref{eq:obj-std} exhibits a \emph{weighted water-filling} structure, analogous to generalized
    water-filling solutions for constrained resource allocation~\cite{Palomar2005Practical}.
    A single dual price coordinates all users, but their water levels differ according to their cache sizes and marginal-utility curves.
    Each user's allocation fills the positive gap between its current progress $x_i$ and its price-dependent water level.

    \begin{theorem}[Weighted Water-Filling] \label{thm:closed-form}
        Under Assumption~\ref{asm:operable}, if $B\ge B_{\min}$, the primal optimum is
        \begin{equation}\label{eq:opt-solution}
        y_i^\star = \max\{x_i,W_i(\lambda^\star)\},
        \qquad b_i^\star=\frac{y_i^\star-x_i}{\Delta t},
        \end{equation}
        where 
    \begin{equation}
    \label{eq:general-waterlevel}
    W_i(\lambda)
    \triangleq
    \begin{cases}
    1,
    & s_i(\lambda)\le A_i'(1),\\[1mm]
    (A_i')^{-1}\bigl(s_i(\lambda)\bigr),
    & A_i'(1)<s_i(\lambda)<A_i'(\tau_i),\\[1mm]
    \tau_i,
    & s_i(\lambda)\ge A_i'(\tau_i).
    \end{cases}
    \end{equation}
    Here, we set $s_i(\lambda)\triangleq\lambda L_i/\Delta t$, and $\lambda^\star\ge0$ is a dual-optimal variable associated with the backhaul-rate constraint.
    When the budget binds, $\lambda^\star$ can be chosen to satisfy
        \begin{equation}\label{eq:budget-eq}
        \sum_i \frac{L_i(y_i^\star-x_i)}{\Delta t}=B.
        \end{equation}
    \end{theorem}
    
    Since $A_i$ is strictly concave on $[\tau_i,1]$, its derivative is strictly decreasing on this interval, and the inverse in the second branch is well defined.
    These optimal solutions are derived via the Karush--Kuhn--Tucker (KKT) conditions of Problem~\eqref{eq:obj-std}; a formal proof is provided in Appendix~\ref{app:proof}.
    
    Let $z_i=L_i(y_i-x_i)$ denote the number of additional cache bits transmitted for user $i$. For every user satisfying $\tilde{x}_i<y_i^\star<1$,
    \begin{equation}
    \label{eq:marginal-equalization}
    \left.
    \frac{\mathrm{d}A_i(x_i+z_i/L_i)}{\mathrm{d}z_i}
    \right|_{z_i=z_i^\star}
    =
    \frac{A_i'(y_i^\star)}{L_i}
    =
    \frac{\lambda^\star}{\Delta t},
    \end{equation}
    where $z_i^\star=L_i(y_i^\star-x_i)$.
    Thus, the optimal allocation equalizes the marginal accuracy gain per additional transmitted cache bit across all interior users.

    Theorem~\ref{thm:closed-form} does not require a particular parametric utility family.
    Specializing the result to a given family requires only evaluating $(A_i')^{-1}$ for the interior branch, while the boundary cases are handled by~\eqref{eq:general-waterlevel}.
    We illustrate this using the algebraic and logistic sigmoids characterized in Section~\ref{sec:utility}.

    \begin{example}[Water Levels for Sigmoid Utilities]
    \label{ex:waterlevels}
    For the interior branch of~\eqref{eq:general-waterlevel}, the algebraic sigmoid in~\eqref{eq:algebraic} gives
    
    \begin{equation}\label{eq:waterlevel-algebraic}
    (A_i')^{-1}\bigl(s_i(\lambda)\bigr) = \tau_i+\frac{1}{k_i} \sqrt{\left(\frac{M_i k_i}{2 s_i(\lambda)}\right)^{2/3}-1}.
    \end{equation}
    For the logistic sigmoid $A_i(y)=M_i/\left(1+e^{-k_i(y-\tau_i)}\right)$, it gives
    \begin{equation}\label{eq:waterlevel-logistic}
    (A_i')^{-1}\bigl(s_i(\lambda)\bigr) = \tau_i+\frac{1}{k_i}\operatorname{arcosh}\left(\frac{M_i k_i}{2 s_i(\lambda)}-1\right).
    \end{equation}
    \end{example}
    
    We adopt the algebraic specialization because it combines strong empirical fit with efficient per-user allocation: once the dual price is determined, each user's optimal allocation can be evaluated using only arithmetic operations and radicals.
    When $B<B_{\min}$, the feasibility condition of Lemma~\ref{lem:feasibility} fails and the closed form above no longer applies; Section~\ref{subsec:importance-admission} handles this regime with an explicit admission policy.   
    
    \subsection{Fallback Policy and Resource Allocation Algorithm}
    
    \label{subsec:importance-admission}
    By Lemma~\ref{lem:feasibility}, Problem~\eqref{eq:obj-std} is feasible if and only if $B\ge B_{\min}$.
    When $B<B_{\min}$, no allocation can bring every active user into its concave region within the current slot.
    The scheduler must therefore invoke a fallback admission policy that prioritizes users under the insufficient backhaul budget.
    The benefit of such an explicit policy is evaluated empirically in Section~\ref{subsec:exp-sensitivity} (Fig.~\ref{fig:exp-ablation}(b)).
    As our default fallback, we adopt \emph{equalized bytes (EB)}.

    Let $\mathcal{S}\triangleq \{i\in\mathcal{N}_t:x_i<\tau_i\}$ denote the users that have not yet reached their concavity anchors. 
    EB equalizes the number of cache bits delivered during the slot among the users in $\mathcal{S}$.
    Since all users share the same slot duration, this is equivalent to initially assigning each user $i\in\mathcal{S}$ a physical backhaul rate of $B/|\mathcal{S}|$, corresponding to the relative allocation $b_i^{\mathrm{EB}}=\frac{B}{|\mathcal{S}|L_i}$.
    Each relative allocation is capped at the rate required to complete the remaining cache transfer within the slot, $(1-x_i)/\Delta t$. 
    If a user reaches this cap, the released backhaul rate is redistributed equally among the remaining unfinished users in $\mathcal{S}$. 
    Users outside $\mathcal{S}$ receive no bandwidth under the fallback policy.
    By prioritizing users that have not yet reached their anchors, EB mitigates sub-$\tau$ starvation under heavy load.

    \begin{algorithm}[!t]
        \caption{Importance-Aware Resource Allocation}
        \label{alg:main}
        \textbf{Input:} Active user index set $\mathcal{N}_{t}$ with user parameters $\{(L_{i}, x_{i}, \tau_{i}, A_{i})\}_{i\in\mathcal N_t}$, total bandwidth $B$, slot duration $\Delta t$. \\
        \textbf{Output:} Bandwidth allocation $\{b_{i}^\star\}_{i \in \mathcal{N}_{t}}$.
        \begin{algorithmic}[1]
            \State $\mathcal{S} \leftarrow \{i \in \mathcal{N}_{t} : x_{i} < \tau_{i}\}$ \Comment{Sub-$\tau$ users}
            \State $B_{\min} \leftarrow \frac{1}{\Delta t}\sum_{i \in \mathcal{S}} L_{i}(\tau_{i} - x_{i})$
            \If{$B < B_{\min}$} \Comment{Infeasibility}
                \State \Return $\textsc{EqualizedBytes}(\mathcal{S}, B, \Delta t)$
            \EndIf
            \State $\lambda^\star \leftarrow \textsc{FindWaterLevel}(\mathcal{N}_{t}, B, \Delta t)$
            \For{each $i \in \mathcal{N}_{t}$}
                \State $y_{i}^\star \leftarrow \max(x_{i}, W_{i}(\lambda^\star))$ \Comment{Optimal target}
                \State $b_{i}^\star \leftarrow (y_{i}^\star - x_{i})/\Delta t$
                \Comment{Bandwidth assignment}
            \EndFor
            \State \Return $\{b_{i}^\star\}_{i \in \mathcal{N}_{t}}$
        \end{algorithmic}        
    \end{algorithm}

    The complete two-regime allocator is summarized in Algorithm~\ref{alg:main}.
    It first evaluates the feasibility threshold $B_{\min}$.
    If $B<B_{\min}$, it invokes the EB fallback.
    Otherwise, it computes the optimal allocation characterized in Theorem~\ref{thm:closed-form} by locating a dual-optimal price $\lambda^\star$ through the bisection subroutine $\textsc{FindWaterLevel}$. 

    \textsc{FindWaterLevel} exploits the monotonicity of the aggregate demand $g(\lambda)\triangleq\sum_{i=1}^{N} L_i\bigl(y_i^\star(\lambda)- x_i\bigr)/\Delta t$, which is continuous and non-increasing in $\lambda$.
    If $g(0)\le B$, the entire remaining cache fits within the slot and $\lambda^\star=0$; otherwise, bisection solves $g(\lambda)=B$ to tolerance $\varepsilon$.
    Each evaluation of $g$ costs $O(|\mathcal{N}_t|)$, and the search range is bounded by $\bar{\lambda}$, the price at which $g(\bar{\lambda})=B_{\min}$, so the allocator runs in $O(|\mathcal{N}_t|\log(\bar{\lambda}/\varepsilon))$ time per slot, i.e., near-linear in the number of active users.

    \subsection{Connection to Classical Water-Filling}  

    When all users have the same utility function and cache size, i.e., $A_i(\cdot)=A(\cdot)$ and $L_i=L$, the user-specific water levels $W_i(\lambda)$ reduce to a common level $W(\lambda)$, and users differ only in their current progress $x_i$.
    When the backhaul constraint is active, let $W^\star\triangleq W(\lambda^\star)$ and $\mathcal{A}\triangleq \{i\in\mathcal{N}_t:x_i<W^\star\}$ denote the set of users receiving positive allocations.    
    The budget equality then gives
    \begin{equation}
        \label{eq:classical-wf}
        W^\star = \frac{1}{|\mathcal{A}|}\sum_{i\in\mathcal{A}} x_i + \frac{B\,\Delta t}{|\mathcal{A}|\,L},
        \quad
        b_i^\star = \frac{1}{\Delta t}\big[W^\star-x_i\big]^+ .
    \end{equation}
    This has the same algebraic form as classical water-filling~\cite{Cover2006elements, Palomar2005Practical}, with the channel-dependent ground level replaced by the progress floor $x_i$ and the total power replaced by the deliverable cache fraction $B\Delta t/L$.
    Then, a user with greater current progress has a higher ground level and requires less additional bandwidth to reach the common target level.   

    \section{Experimental Results}\label{sec:experimental}
    
    \subsection{Experiment Settings}\label{subsec:exp-settings}
    
    \subsubsection{Simulation environment}
    We follow the slotted system model of Section~\ref{sec:system}.
    At the beginning of each time slot of duration $\Delta t$, the scheduler observes all active users $\mathcal{N}_{t}$ and updates the bandwidth allocation, which remains fixed throughout the slot.
    Handover events arrive according to a Poisson process with mean rate $\rho$ (users per second).
    Upon each event, the context length of a new user is uniformly sampled from $\{4\mathrm{K},\,8\mathrm{K},\,16\mathrm{K}\}$ tokens to model heterogeneous context lengths.
    Each user attempts to complete its handover within an allowed transfer window $T_{\max}$.
    Unless otherwise noted, the default settings are $B=\SI{20}{Gbps}$, $\Delta t=\SI{100}{ms}$, $T_{\max}=\SI{500}{ms}$, $\rho=\SI{4}{users/s}$, and context lengths drawn uniformly from $\{4\mathrm{K},8\mathrm{K},16\mathrm{K}\}$, with each point averaged over \SI{100}{s} of simulation across 100 Monte-Carlo runs.

    We adopt $B=\SI{20}{Gbps}$ as a beyond-5G baseline.
    While over-provisioned networks ($\ge$\SI{50}{Gbps}) make allocation trivial, and severely limited ones reduce the problem to pure admission control, \SI{20}{Gbps} serves as a practical yet challenging operating point; its sensitivity is examined in Section~\ref{subsec:exp-main}.
    
    Similarly, the slot duration $\Delta t =$~\SI{100}{ms} balances the re-allocation period against scheduler overhead.
    The transfer window $T_{\max} =$ \SI{500}{ms} is a mobility-imposed budget for completing the transfer in the background, \emph{not} an inference-time latency: in soft-handover architectures, the target can begin loading the cache before the user formally migrates.
    When the window expires, the transfer is truncated, and the user resumes inference on whatever prefix has arrived.
    We analyze the sensitivity of both $\Delta t$ and $T_{\max}$ in Section~\ref{subsec:exp-sensitivity}.

    All experiments use Qwen3-8B~\cite{Yang2025qwen3} in BF16 precision with KV cache importance scored by Fast KVzip~\cite{Kim2026fastkvzip}.
    The per-user accuracy utility $A_i(\cdot)$ follows the algebraic sigmoid form of Eq.~\eqref{eq:algebraic}, with parameters $M_i, k_i, \tau_i$ fitted independently for each context length $\{4, 8, 16\}$K on the RULER benchmark~\cite{Hsieh2024ruler} as detailed in Section~\ref{sec:utility}.
    Because delivery proceeds in importance order, the received fraction $x_i$ in the simulation indexes exactly the top-$x_i$ fraction of Section~\ref{subsec:system-ordering}.
    Thus, the reported accuracy at each slot is obtained by evaluating this fitted sigmoid at the user's current $x_i$.
    To further validate that this sigmoid-based simulation tracks real inference under the scheduler-induced distribution of $x_i$, we also measure end-to-end accuracy under the same simulation settings.
    It is obtained by having each user perform inference using the directly delivered top-$x_i$ fraction of the cache.

    \subsubsection{Baselines and metrics}
    Throughout, \emph{Ours} refers to the proposed two-regime allocator.
    When the per-slot budget suffices to lift every sub-$\tau$ user to the concavity anchor (i.e., the feasibility regime), it applies the weighted water-filling of Theorem~\ref{thm:closed-form}.
    Otherwise, it falls back to an admission policy that distributes the budget equally among the contending sub-$\tau$ users.
    We compare the proposed allocator against three baselines, each of which applies a fixed rule in every slot, irrespective of feasibility.
    \begin{itemize}
        \item \textbf{Equal allocation (EA)}: The backhaul capacity is distributed uniformly across active users so that each user receives bandwidth $L_{i}b_{i} = B / |\mathcal{N}_{t}|$~\SI{}{bps}, independent of importance or current progress.
    
        \item \textbf{Winner-take-all (WTA)}: The greedy heuristic which assigns the slot budget to the user with the largest single-slot marginal accuracy gain. 
        We use the \emph{cascading} variant throughout: when the chosen user cannot absorb the full slot budget (its remaining cache is smaller than $B \cdot \Delta t$), the leftover rolls over to the next-best user within the same slot.
        The pure WTA, which discards any leftover capacity, is strictly dominated and is therefore omitted from the main comparisons.

        \item \textbf{Proportional-fair (PF)}: The \emph{unweighted} proportional-fair allocation that maximizes $\sum_i \log(L_i y_i)$ subject to the shared budget, without accessing the accuracy curves~\cite{Kelly1998rate}.
        In our single-resource setting, its KKT solution reduces to a common-water-level water-filling on cumulative received bits, which we solve directly. 
        It thus represents a fairness-oriented baseline that is agnostic to the sigmoidal accuracy utility.
    \end{itemize}
    
    These baselines span the fairness--throughput spectrum of resource allocation, with EA and WTA at its two extremes and PF occupying the middle ground~\cite{Kelly1998rate}.
    All schemes share the same importance ordering, differing only in allocation: among the baselines, only WTA consults the accuracy utility, while EA and PF are utility-agnostic.
    Additionally, we compare the proposed method against compute-based baselines: a target-side re-prefill strategy and the hybrid token and cache transmission design from~\cite{Lee2026low-latency}.

    \subsection{Main Results} \label{subsec:exp-main}
    \begin{figure}[!t]
        \centering
        \subfloat[Accuracy vs.  $\rho$.]{%
            \includegraphics[width=0.235\textwidth]{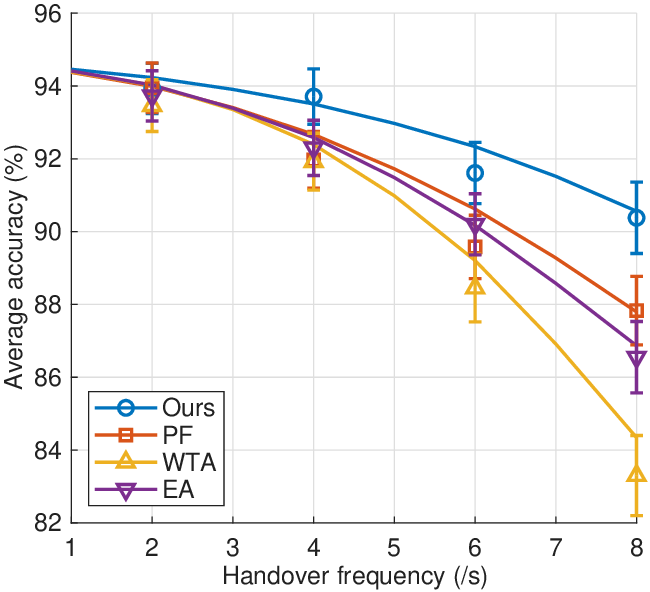}
            \label{fig:exp-rho}}
        \hfill
        \subfloat[Accuracy vs. $B$.]{%
            \includegraphics[width=0.235\textwidth]{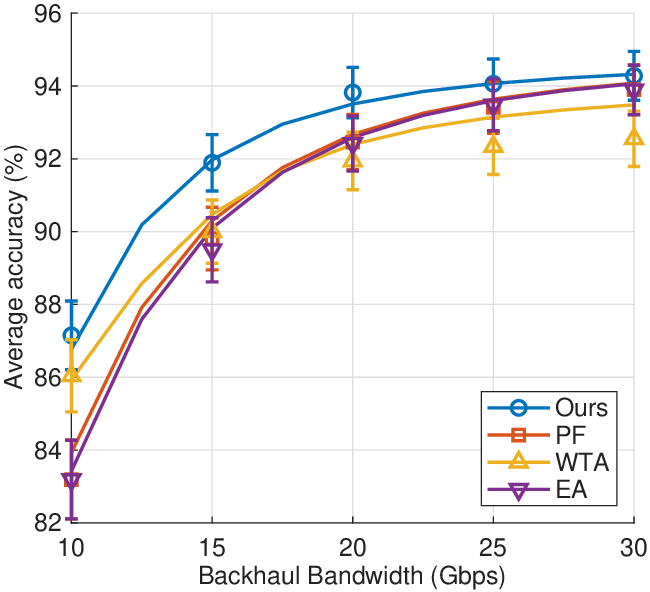}
            \label{fig:exp-bandwidth}}
        \caption{Performance comparison of the proposed allocator against EA, PF, and WTA baselines under edge-LLM handover conditions across (a) varying handover frequency $\rho$ and (b) varying backhaul bandwidth $B$.
        Solid curves are computed from the simulation, evaluating each user's accuracy at its delivered fraction $x_i$ through the fitted sigmoid utility $A_i(x_i)$.
        Markers additionally show end-to-end measured accuracy.
        Each marker averages the first $4000$ users of the trace, and error bars are 95\% confidence intervals over RULER samples.
        The close agreement between curves and markers confirms that the sigmoid-based simulation reliably predicts real inference accuracy.
        }
        \label{fig:exp-measured}
        \vspace{-4mm}
    \end{figure}
    
    We first evaluate the performance of the proposed allocator against the baselines along two axes that jointly govern the system load: the handover frequency $\rho$ and the backhaul bandwidth $B$. 
    In Fig.~\ref{fig:exp-measured}, the solid curves denote the simulation accuracy derived from the fitted sigmoid utility in Fig.~\ref{fig:fit}, whereas the markers with error bars indicate the measured end-to-end accuracy.
    
    Focusing on Fig.~\ref{fig:exp-measured}(a), $\rho$ is varied over \qtyrange{1}{8}{users/s}.
    As $\rho$ increases, the accuracy of all schemes decreases, since the fixed budget $B$ must be shared among more concurrent transfers.
    The proposed allocator outperforms the baselines throughout, and the gap widens, since heavy loads make importance-aware prioritization essential for maximizing accuracy with scarce resources.
    At the default $\rho=\SI{4}{users/s}$, it attains \SI{93.7}{\%} accuracy, within \SI{0.5}{pp} of the \SI{94.1}{\%} full-cache ceiling. 

    In Fig.~\ref{fig:exp-measured}(b), we sweep $B$ over~\qtyrange{10}{30}{Gbps}.
    The result demonstrates that the proposed allocator consistently outperforms all three baselines across the entire operating range.
    Among the baselines, the two extremes, EA and WTA, exhibit a characteristic crossover.
    In the bandwidth-scarce regime, WTA outperforms EA by concentrating the limited budget to ensure at least one user reaches the operable region, whereas EA leaves all users starved.
    Conversely, in the bandwidth-abundant regime, the ordering reverses: EA benefits from serving multiple users in parallel, while WTA wastes capacity by driving its current winner deep into the saturated tail of its utility before advancing to the next.
    PF weights each share by progress and cache size, yet trails the proposed allocator: blind to the accuracy curve, it over-serves saturated users and under-serves those near $\tau$.
    
    \subsection{Sensitivity and Ablation Analysis} \label{subsec:exp-sensitivity}
    
    \begin{figure}[t]
        \centering
        \subfloat[Accuracy vs. $\Delta t$.]{%
            \includegraphics[width=0.235\textwidth]{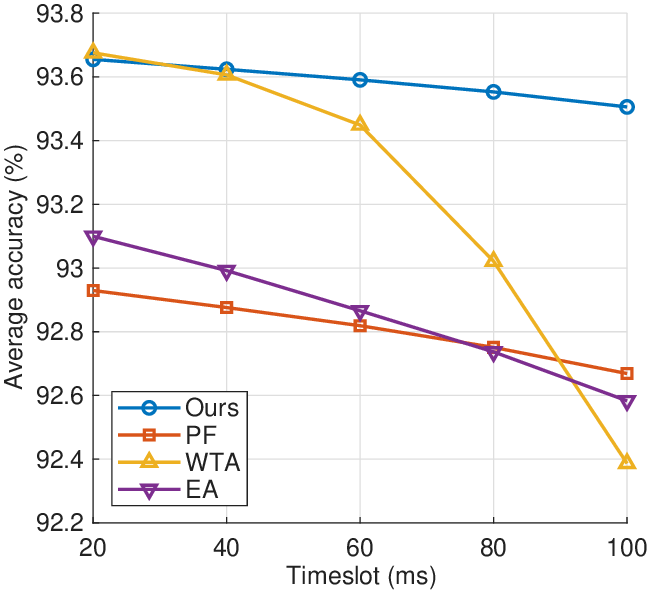}
            \label{fig:exp-timeslot}}
        \subfloat[Accuracy vs. $T_{\max}$.]{%
            \includegraphics[width=0.235\textwidth]{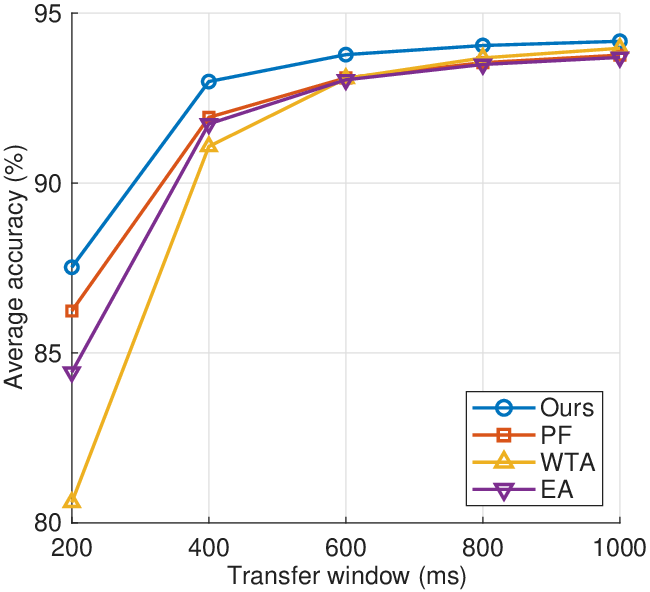}
            \label{fig:exp-window}}    
        \caption{
        Performance comparison of the proposed allocator against baselines across (a) varying slot duration $\Delta t$ and (b) varying transfer window $T_{\max}$.
        }
        \label{fig:exp-sensitivity}
        \vspace{-4mm}
    \end{figure}

    We next examine robustness to two system configuration parameters: the slot duration $\Delta t$, and the transfer window $T_{\max}$.
    Fig.~\ref{fig:exp-sensitivity}(a) sweeps $\Delta t$ over $20$--$100$\,ms with all other parameters at their defaults.
    The proposed allocator stays above all three baselines across the entire range. EA and WTA cross near $\Delta t\approx$~\SI{90}{ms}: as the slot gets coarser, the better fixed rule switches from one to the other.
    The proposed allocator, however, keeps its accuracy even at the coarsest slots, so it can re-allocate far less often without losing quality.

    In Fig.~\ref{fig:exp-sensitivity}(b), the transfer window is swept over \qtyrange{200}{1000}{ms}, showing the same crossover near \SI{400}{ms}: short windows penalize EA, while longer ones let its parallelism overtake WTA.
    The proposed allocator again outperforms throughout, so its advantage is not tied to a particular $T_{\max}$.

    \begin{figure}[!t]
        \centering
        \subfloat[Effect of ordering.]{%
            \includegraphics[width=0.235\textwidth]{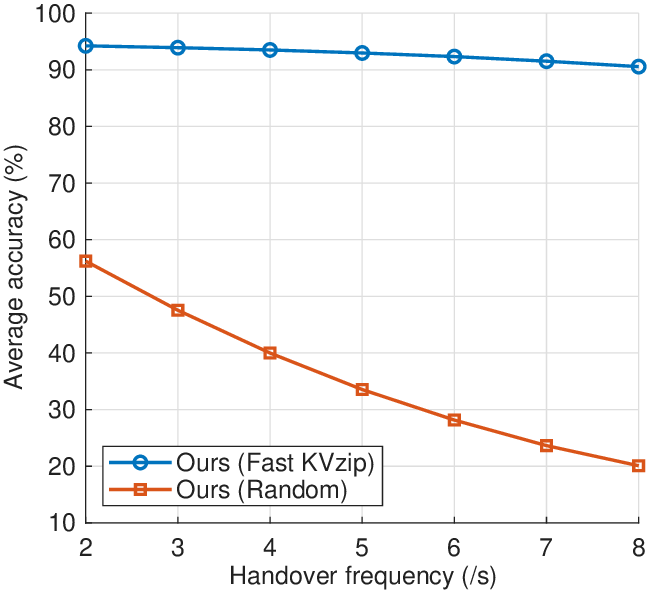}
            \label{fig:exp-ordering}}
        \hfill
        \subfloat[Admission control policy.]{%
            \includegraphics[width=0.235\textwidth]{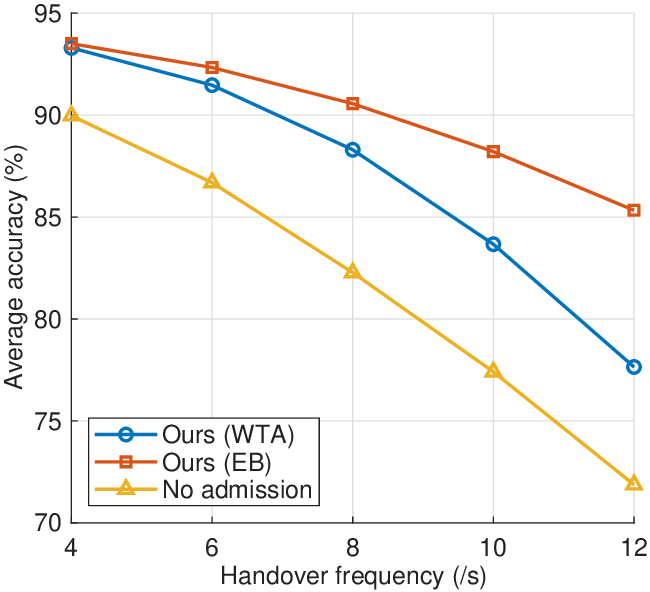}
            \label{fig:exp-admission}}

        \caption{The impacts of (a) importance-aware KV cache ordering relative to random ordering and (b) the admission control mechanism.
        }
        \label{fig:exp-ablation}
        \vspace{-4mm}
    \end{figure}

    To analyze the sources of the performance improvements, we isolate the two core mechanisms behind our method in Fig.~\ref{fig:exp-ablation}: importance-based cache ordering and the admission-control policy.
    For the latter, we compare the default policy (\emph{Ours (EB)}, which refers to equalized bytes) against the same allocator with cascading WTA as its fallback (\emph{Ours (WTA)}).
    Ours (WTA) is distinct from the standalone WTA baseline used above, which applies cascading WTA in every slot.
    
    Fig.~\ref{fig:exp-ablation}(a) sweeps $\rho$ for the proposed allocator under Fast KVzip and random ordering.
    Importance ordering keeps average accuracy above \SI{90}{\%} across the entire range, whereas random ordering falls sharply from \SI{56.7}{\%} at $\rho=2$ to \SI{20.2}{\%} at $\rho=8$.
    This demonstrates that importance-based ordering strengthens ImpactHO's anytime inference capability by transmitting higher-ranked cache entries first, thereby improving inference accuracy at intermediate transfer points.
    
    Fig.~\ref{fig:exp-ablation}(b) sweeps $\rho$ over \qtyrange{4}{12}{users/s} to stress the infeasibility regime.
    Both variants with admission control are compared against the same allocator with the fallback disabled (\emph{No admission}).
    At low $\rho$, the three curves coincide, as the feasibility regime dominates.
    As $\rho$ grows, both variants with admission control pull away sharply from the one without it, confirming that admission control is the dominant factor in this regime.
    We define the sub-$\tau$ starvation rate as the fraction of users whose received fraction never reaches $\tau$ within the transfer window $T_{\max}$.
    At $\rho=\SI{12}{users/s}$, the sub-$\tau$ starvation rate is \SI{20.8}{\percent} without admission versus \SI{0.65}{\percent} with EB.
    At our main slot duration ($\Delta t = \SI{100}{ms}$), EB consistently outperforms WTA, whose starvation rate is \SI{15.5}{\percent} versus EB's \SI{0.65}{\percent}.
    This is because WTA drives one user into the saturated tail beyond $\tau$ before advancing the next, whereas EB lets several progress toward operability.
    
    \subsection{Comparison with Compute-Based Baselines}
    \label{subsec:exp-comparison}
    \begin{table}[t]
    \centering
    \footnotesize 
    \renewcommand{\arraystretch}{1.2} 
    \setlength{\tabcolsep}{4pt}
    \caption{Per-user latency comparison of the compute-based baselines and the proposed method for Qwen3-8B.}
    \label{tab:user_latency}
    \begin{tabular}{|c|c|c|c|c|c|}
        \hline
        \multirow{2}{*}{Context}    & Re-prefill (ms) & Hybrid (ms) & Ours (ms) \\ \cline{2-4} 
          & NVIDIA RTX PRO 6000 & PRO 6000, \SI{20}{Gbps}   & \SI{20}{Gbps}  \\
        \hline\hline
        4K       & 300    & 211      & 229   \\ \hline
        8K       & 676     & 396      & 353    \\ \hline
        16K      & 1,566   & 694    & 590    \\ \hline\hline
        \textbf{Avg} & \textbf{847} & \textbf{434} & \textbf{391} \\
        \hline
    \end{tabular}
    \vspace{-4mm}
    \end{table}
    We evaluate the per-user latency performance of the proposed method by comparing it against compute-based re-prefill baselines.
    In Table~\ref{tab:user_latency}, \emph{Re-prefill} reconstructs the context directly from the raw tokens, whereas \emph{Hybrid} transmits both tokens and KV cache to balance the prefill processing time and cache transmission latency, a scheme appropriately modified from~\cite{Lee2026low-latency} for our evaluation.
    In the hybrid baseline, KV cache importance is not taken into account, and the backhaul bandwidth is allocated entirely in a first-come, first-served (FCFS) manner.
    Because the data size of the tokens is negligible compared to the KV cache, we do not consider token transmission latency.
    
    For comparison, we generate user arrivals following a Poisson process with $\rho=4$ over the time window $t\in[0,1]$\SI{}{s}, and repeat this process 1,000 times.
    For each user, the service latency is measured from the moment of request arrival.
    Specifically, for the re-prefill and hybrid baselines, latency spans until the context is fully restored; for the proposed method, it lasts until a sufficient amount of the KV cache is transferred to reach \SI{99}{\percent} of the full accuracy.
    Here, the transfer window $T_{\max}$ is set to $\infty$, so that no scheme is truncated by the deadline.
    We assume that both the re-prefill and hybrid baselines are provisioned with sufficient GPUs to process each user request independently, so their latencies are free of compute contention.
    This idealized assumption establishes a much stricter baseline than a direct comparison with~\cite{Lee2026low-latency}.
    
    At a context length of 4K, the proposed method shows comparable performance to the baselines.
    However, as the context length increases, the proposed method outperforms the baselines, which suffer from severe computational overhead during full context recomputation.
    Moreover, the proposed method can further reduce latency by leveraging its anytime property, while ensuring graceful performance degradation.
    
    \subsection{Discussion} \label{subsec:exp-discussion}
    The reported gains require importance-ordered transmission, which adds two source-side costs and a small wire overhead.
    (i) \emph{Scoring}: assigning a Fast KVzip~\cite{Kim2026fastkvzip} score to every KV entry is a one-time, per-session computation done while the source still serves the user, not at handover.
    (ii) \emph{Sorting}: reordering by score is an $O(n\log n)$ sort over indices, a small fraction of the transfer time.
    (iii) \emph{Metadata}: because the transfer may be truncated, each entry carries its 3-byte $(\text{layer},\text{head},\text{token})$ coordinate against the \SI{0.5}{KB} payload of Section~\ref{subsec:system-ordering}, under \SI{1}{\%} overhead.

    The proposed allocator optimizes each slot rather than the horizon.
    To bound the resulting loss, we compare against a clairvoyant upper bound that observes all future arrivals and maximizes the objective without the concave-region restriction.
    Solved to certified global optimality via a piecewise-linear MILP with a validated upper correction, this bound never underestimates the true continuous optimum.
    The proposed allocator attains \qtyrange{98.2}{99.5}{\%} of it across all loads (within \SI{0.53}{\%} at $\rho=1$, \SI{1.77}{\%} at $\rho=8$).
    Thus, neither the myopic objective nor the concave-region restriction costs much in practice, though the widening gap at higher load suggests horizon-aware scheduling as a direction for future work.
    
    \section{Conclusion}\label{sec:conclusion}
    
    We presented the ImpactHO framework, an importance-aware framework for maximizing the average inference accuracy across users during edge LLM handover subject to limited backhaul bandwidth. 
    ImpactHO combines three key components:
    (1) importance-ordered sequential KV cache transfer, which exploits sparsity in token-level cache importance to transmit the most useful entries first;
    (2) an empirically validated sigmoid characterization of inference accuracy with partially transferred caches; and
    (3) an optimal weighted water-filling allocator whose homogeneous special case reduces to classical water-filling.
    In realistic edge regimes, importance-ordered transfer attains near-full-cache accuracy with lower latency than baselines.
    Moreover, across a range of slot durations, backhaul bandwidths, and concurrent handover loads, the proposed allocator consistently achieves higher average accuracy than the baselines. 
    These results confirm that jointly accounting for KV cache importance and backhaul resource allocation enables accurate and efficient LLM handover at the network edge.

\appendices
\section{Proof of Theorem~\ref{thm:closed-form}}
\label{app:proof}

    Recall the effective lower bound $\tilde{x}_i = \max(\tau_i,\, x_i)$ from Section~\ref{sec:system}, which collapses the box constraint of~\eqref{eq:obj-std} to the single interval $y_i \in [\tilde{x}_i, 1]$.
    We work in the feasibility regime $B \ge B_{\min}$ of Lemma~\ref{lem:feasibility}.
    By Assumption~\ref{asm:operable}, $-A_i$ is strictly convex on $[\tau_i, 1]$ and the constraint set is a polytope, so Problem~\eqref{eq:obj-std} is convex with a strictly convex objective.
    The Lagrangian is
    \begin{align}
        \mathcal{L}(\mathbf{y}, \lambda, \boldsymbol{\mu}, \boldsymbol{\nu})
        =& -\sum_{i=1}^{N} A_i(y_i) + \lambda \!\left( \sum_{i=1}^{N} \frac{L_i(y_i - x_i)}{\Delta t} - B \right) \nonumber \\
        & + \sum_{i=1}^{N} \mu_i (\tilde{x}_i - y_i) + \sum_{i=1}^{N} \nu_i (y_i - 1).
        \label{eq:full-lagrangian}
    \end{align}
    The KKT conditions for the optimal 
    $(\mathbf{y}^\star, \lambda^\star, \boldsymbol{\mu}^\star, \boldsymbol{\nu}^\star)$ are as 
    follows:
    \begin{align}
        & \sum_{i=1}^{N} \frac{L_i(y_i^\star - x_i)}{\Delta t} \le B, \; \lambda^\star \ge 0,
        \label{eq:kkt-1} \\
        & \lambda^\star \!\left( \sum_{i=1}^{N} \frac{L_i(y_i^\star - x_i)}{\Delta t} - B \right) = 0,
        \label{eq:kkt-2} \\
        & \tilde{x}_i \le y_i^\star \le 1, \; \mu_i^\star, \nu_i^\star \ge 0,
        \label{eq:kkt-3} \\
        & \mu_i^\star (\tilde{x}_i - y_i^\star) = 0, \; \nu_i^\star (y_i^\star - 1) = 0
        \label{eq:kkt-4}
    \end{align}
    for $i \in \{1, \ldots, N\}$.
    From $\partial \mathcal{L}/\partial y_i = 0$, the stationarity condition gives
    \begin{equation}
        A_i'(y_i^\star) = \frac{\lambda^\star L_i}{\Delta t} - \mu_i^\star + \nu_i^\star.
        \label{eq:foc}
    \end{equation}

    By Assumption~\ref{asm:operable}, $A_i'$ is strictly decreasing on $[\tau_i,1]$, hence invertible.
    Stationarity~\eqref{eq:foc} with complementary slackness~\eqref{eq:kkt-4} gives the three branches of~\eqref{eq:general-waterlevel}: the interior case $\mu_i^\star=\nu_i^\star=0$ yields $A_i'(y_i^\star)=s_i(\lambda^\star)$, while $\nu_i^\star>0$ and $\mu_i^\star>0$ activate the bounds $W_i=1$ and $W_i=\tau_i$ when $s_i(\lambda^\star)\le A_i'(1)$ and $s_i(\lambda^\star)\ge A_i'(\tau_i)$, respectively.
    Combining with $y_i^\star\ge x_i$ from $b_i^\star\ge0$ gives
    \begin{equation}
        y_i^\star \;=\; \max\{x_i,\, W_i(\lambda^\star)\}.
        \label{eq:y-star-eff}
    \end{equation}
    It remains to characterize $\lambda^\star$.
    Define the aggregate demand
    \begin{equation}
        g(\lambda) \;\triangleq\; \sum_{i=1}^{N} \frac{L_i\, (y_i^\star(\lambda) - x_i)}{\Delta t}
        \label{eq:demand-fn}
    \end{equation}
    with $y_i^\star(\lambda)$ given by~\eqref{eq:y-star-eff}.
    Because $A_i'$ is strictly decreasing on $[\tau_i, 1]$, $W_i(\lambda)$ is strictly decreasing in $\lambda$ on its interior region and constant elsewhere; taking $\max\{x_i,\cdot\}$ preserves this monotonicity, so $g(\lambda)$ is continuous and non-increasing on $\lambda \ge 0$, with boundary values $g(0) = \sum_i L_i (1 - x_i)/\Delta t$ and $g(\bar{\lambda}) = B_{\min}$, where $\bar{\lambda} \triangleq \max_i A_i^\prime (\tau_i)\Delta t/L_i$.
    By the intermediate value theorem, a $\lambda^\star\ge0$ with $g(\lambda^\star)=B$ exists for every $B\in[B_{\min}, g(0)]$, and for $B> g(0)$ the budget is slack with $\lambda^\star=0$; $\lambda^\star$ is unique when some user is strictly interior, and otherwise any such $\lambda^\star$ yields the same $\mathbf{y}^\star$, which is unique by strict convexity.
    This completes the proof of Theorem~\ref{thm:closed-form}.

\bibliographystyle{IEEEtran}
\bibliography{abrv,mybib}

\end{document}